\documentclass{article}
\usepackage{kr}

\usepackage{times}
\usepackage{soul}
\usepackage{url}
\usepackage[hidelinks]{hyperref}
\usepackage[utf8]{inputenc}
\usepackage[small]{caption}
\usepackage{graphicx}
\usepackage{amsmath}
\usepackage{amsthm}
\usepackage{booktabs}
\usepackage{algorithm}
\usepackage{algorithmic}
\usepackage{balance}

\usepackage{amssymb}
\usepackage{amsfonts}
\usepackage{dsfont}
\usepackage{cleveref}
\usepackage{stmaryrd}
\usepackage{scalefnt}
\usepackage{tikz}

\usepackage{mathtools}

\usepackage{relsize}
\usepackage{enumitem}
\usepackage{thmtools}
\usepackage{stackengine}

\newtheorem{example}{Example}

\newtheoremstyle{mystyle}%
{}{}%
{\itshape}%
{}%
{\bfseries}
{.}%
{.5em}%
{\indent\thmname{#1}\thmnumber{ #2}\thmnote{ (#3)}}

\declaretheorem[style=theorem]{definition}

\newcommand{\HyperSL}{HyperSL}
\newcommand{\HyperCTLS}{HyperCTL$^*$}
\newcommand{\SLI}{SL$_{\mathit{ii}}$}
\newcommand{\SL}{SL}

\newcommand{\bA}{{\vec{\alpha}}}
\newcommand{\bX}{{\vec{x}}}

\newcommand{\stratVars}{\mathcal{X}}

\DeclareMathOperator{\ltlN}{\normalfont\textsf{X}}
\DeclareMathOperator{\ltlG}{\normalfont\textsf{G}}
\DeclareMathOperator{\ltlF}{\normalfont\textsf{F}}

\DeclareMathOperator{\ltlU}{\normalfont\textsf{U}}
\DeclareMathOperator{\ltlW}{\normalfont\textsf{W}}

\newcommand{\ap}{\mathit{AP}}
\newcommand{\pathVars}{\mathcal{V}}

\newcommand{\ldot}{\mathpunct{.}}

\newcommand{\nat}{\mathbb{N}}

\newcommand{\calG}{\mathcal{G}}

\newcommand{\calO}{\mathcal{O}}

\newcommand{\obs}{\mathit{Obs}}

\newcommand{\agents}{\mathit{{Agts}}}
\newcommand{\moves}{\mathbb{A}}

\newcommand{\play}[0]{\mathit{Play}}
\newcommand{\strats}[1]{\mathit{Str}(#1)}

\newcommand{\clanglei}[2]{ #1 \textbf{@} #2 }
\newcommand{\cparai}[2]{ #1 \textbf{@}  #2 }

\makeatletter
\DeclareRobustCommand\bigop[1]{%
	\mathop{\vphantom{\sum}\mathpalette\bigop@{#1}}\slimits@
}
\newcommand{\bigop@}[2]{%
	\vcenter{%
		\sbox\z@{$#1\sum$}%
		\hbox{\resizebox{\ifx#1\displaystyle.9\fi\dimexpr\ht\z@+\dp\z@}{!}{$\m@th#2$}}%
	}%
}
\makeatother

\makeatletter
\newcommand{\superimpose}[2]{{%
		\ooalign{%
			\hfil$\m@th#1\@firstoftwo#2$\hfil\cr
			\hfil$\m@th#1\@secondoftwo#2$\hfil\cr
		}%
}}
\makeatother

\newcommand{\bigdot}{\DOTSB\bigop{\cdot}}

\newcommand\ScaleExists[1]{\vcenter{\hbox{\scalefont{#1}$\exists$}}}
\newcommand\ScaleForall[1]{\vcenter{\hbox{\scalefont{#1}$\forall$}}}

\DeclareMathOperator*\bigexists{%
	\vphantom\sum
	\mathchoice{\ScaleExists{1.7}}{\ScaleExists{1.4}}{\ScaleExists{1}}{\ScaleExists{0.75}}}

\DeclareMathOperator*\bigforall{%
	\vphantom\sum
	\mathchoice{\ScaleForall{1.7}}{\ScaleForall{1.4}}{\ScaleForall{1}}{\ScaleForall{0.75}}}

\definecolor{mydarkviolet}{HTML}{4A3078}
\newcommand{\agent}[1]{{\color{mydarkviolet}#1}}

\newcommand{\psiSLI}{{\psi}}
\newcommand{\varphiSLI}{{\varphi}}

\newcommand{\psiSL}{{\psi}}
\newcommand{\varphiSL}{{\varphi}}

\newcommand{\transi}[1]{\llbracket #1 \rrbracket}

\newcommand{\slToHyper}[1]{{\llparenthesis} #1 {\rrparenthesis}}

\newcommand{\bindi}[2]{(#1\!\blacktriangleright\!#2)}

\newif\iffullversion
\fullversiontrue

\newcommand{\ifFull}[2]{\iffullversion#1\else#2\fi}

\title{Strategy Logic, Imperfect Information, and Hyperproperties}

\author{%
Raven Beutner\and Bernd Finkbeiner\\
\affiliations
CISPA Helmholtz Center for Information Security, Germany
}

\begin{document}

\maketitle

\begin{abstract}
  Strategy logic (SL) is a powerful temporal logic that enables first-class reasoning over strategic behavior in multi-agent systems (MAS).
  In many MASs, the agents (and their strategies) cannot observe the global state of the system, leading to many extensions of SL centered around imperfect information, such as \emph{strategy logic with imperfect information} (SL$_\mathit{ii}$).
  Along orthogonal lines, researchers have studied the combination of strategic behavior and hyperproperties. 
  Hyperproperties are system properties that relate multiple executions in a system and commonly arise when specifying security policies. 
 Hyper Strategy Logic (HyperSL) is a temporal logic that combines quantification over strategies with the ability to express hyperproperties on the executions of different strategy profiles. 
	In this paper, we study the relation between SL$_\mathit{ii}$ and HyperSL.
  Our main result is that both logics (restricted to formulas where no state formulas are nested within path formulas) are \emph{equivalent} in the sense that we can encode SL$_\mathit{ii}$ instances into HyperSL instances and vice versa. 
  For the former direction, we build on the well-known observation that imperfect information \emph{is} a hyperproperty. 
  For the latter direction, we construct a self-composition of MASs and show how we can simulate hyperproperties using imperfect information. 
\end{abstract}

\section{Introduction}

Multi-agent systems (MAS) are ubiquitous in our everyday lives, necessitating the need for formal guarantees on their behavior. 
In MASs, we typically reason about the ability of groups of agents, which requires reasoning about \emph{strategies}.
This led to the development of powerful temporal logics like ATL/ATL$^*$ \cite{AlurHK02} and Strategy Logic (SL) \cite{ChatterjeeHP10,MogaveroMPV14}. 
While the former can implicitly reason about strategic ability (e.g., a group of agents has \emph{some} strategy to enforce a certain goal), the latter features \emph{explicit} quantification over strategy, allowing the same strategy to be used in multiple contexts, which is critical to express important properties like Nash equilibria.

\paragraph{Strategy Logic With Imperfect Information}

In plain SL, the strategies of agents can observe the entire state of the MAS \cite{MogaveroMPV14}.
In most models of real-world situations, this is unrealistic, i.e., an agent typically acts on some local sensing ability and must thus act under incomplete information. 
This observation led to many logics that can reason about strategic behavior under imperfect information \cite{abs-1908-02488,BelardinelliLMR17,BerthonMM17}. 
A particularly powerful logic among these is SL with imperfect information (\SLI{}) \cite{BerthonMMRV17}, which extends SL with the ability to quantify over strategies with a given observation model.
For example, 
\begin{align*}
	\exists x^o. \forall y^{o'}. \bindi{\agent{1}}{x} \bindi{\agent{2}}{y} \bindi{\agent{3}}{y} \ltlG \ltlF \mathit{goal}
\end{align*}
expresses that there exists some strategy with observation model $o$ (formally, $o$ is associated with an indistinguishability relation on states of the MAS), such that for every strategy $y$ under observation $o'$, the play where agent $\agent{1}$ plays strategy $x$, and agents $\agent{2}$ and $\agent{3}$ play $y$ (i.e., the play under strategy profile $(\agent{1} \mapsto x, \agent{2} \mapsto y, \agent{3} \mapsto y)$) satisfies $ \ltlG \ltlF \mathit{goal}$.

\paragraph{Hyper Strategy Logic}

Along orthogonal lines, SL has been extended with the concept of \emph{hyperproperties} \cite{ClarksonS08}, i.e., properties that relate multiple paths in a system. 
In plain SL, we can use the same strategy in different situations, but each strategy profile is evaluated against an LTL formula.
As a result, we can only express properties on individual strategy profiles and then reason about Boolean combinations of these properties. 
For many properties (e.g., security and robustness policies), we need to compare \emph{multiple} executions to, e.g., see how different high-security inputs impact the low-security observations of a system.
HyperSL \cite{BeutnerF24SL}, extends SL with the concept of path variables (similar to logics like \HyperCTLS{} \cite{ClarksonFKMRS14}).
For example, 
\begin{align*}
	\exists x. \forall y. \forall z. \big((\neg \mathit{g}_{\pi_2}) \ltlU \mathit{g}_{\pi_1} \big) \left[\begin{aligned}
		\pi_1 : (\agent{1} \mapsto x, \agent{2} \mapsto x, \agent{3} \mapsto z)\\
		\pi_2 : (\agent{1} \mapsto x, \agent{2} \mapsto y, \agent{3} \mapsto z)
	\end{aligned}\right]
\end{align*}
states that there exists some strategy $x$ (under full information), such that for all $y, z$, the strategy profile $(\agent{1} \mapsto x, \agent{2} \mapsto x, \agent{3} \mapsto z)$ reaches a goal $\mathit{g}$ at least as fast as profile $(\agent{1} \mapsto x, \agent{2} \mapsto y, \agent{3} \mapsto z)$. 
To express this, we construct \emph{two} paths $\pi_1, \pi_2$ using different strategy profiles, and, within the LTL body, can refer to atomic propositions on both paths. 

\paragraph{The Connection}

At first glance, \SLI{} and \HyperSL{} appear orthogonal.
The former is centered around the information of agents, whereas the latter is targeted at quantitative properties and security policies. 
However, in this paper, we show that both logics (restricted to formulas where no state formulas are nested within path formulas) are equally expressive in the sense that we can translate model-checking instances between both. 
This connection allows for a unified study of (imperfect) knowledge and hyperproperties, and potentially allows the transfer of decidability results between \SLI{} and \HyperSL{} (model-checking is undecidable for both logics). 
To translate \SLI{} into \HyperSL{} (\Cref{sec:enc1}), we build on the well-known observation that a strategy acting under imperfect knowledge \emph{is} a hyperproperty, i.e., the strategy should pick the same action on all \emph{pairs} of executions that appear indistinguishable. 
For the latter direction (\Cref{sec:hypersl_in_sli}), we simulate hyperproperties within \SLI{}. 
Our key observation here is that we can construct the self-composition of an MAS and, using imperfect information, simulate multiple executions within this composition.  

\section{Preliminaries}\label{sec:prelim}
We let $\ap$ be a fixed finite set of atomic propositions and fix a finite set of agents $\agents = \{\agent{1}, \ldots, \agent{n}\}$.

\paragraph{Concurrent Game Structures}

As the underlying model of MASs, we use concurrent game structures (CGS) \cite{AlurHK02}.
A CGS is a tuple $\calG = (S, s_0, \moves, \kappa, L)$ where $S$ is a finite set of states, $s_0 \in S$ is an initial state, $\moves$ is a finite set of actions, $\kappa : S \times (\agents \to \moves)\to S$ is a transition function, and $L : S \to 2^\ap$ is a labeling function.
The transition function takes a state $s$ and an action profile $\vec{\alpha} : \agents \to \moves$ and returns a unique successor state $\kappa(s, \vec{\alpha})$.
A strategy in $\calG$ is a function $f : S^+ \to \moves$, mapping finite plays to actions.
We denote the set of all strategies in $\calG$ with $\strats{\calG}$.
Once we fix a strategy for each agent, we obtain a unique path in the CGS (cf.~\ifFull{the appendix}{the full version \cite{fullVersion}}).

\paragraph{\SLI}

\SLI{} \cite{BerthonMMRV17} extends plain \SL{} by quantifying over strategies with a given observation model. 
Let $\calG = (S, s_0, \moves, \kappa, L)$ be a fixed CGS, and let $\obs$ be a fixed finite set of so-called observations.
An observation family $\{\sim_o\}_{o \in \obs}$ associates each $o \in \obs$ with an equivalence relation $\sim_o \subseteq S \times S$.
For a strategy with observation $o$, two states $s \sim_o s'$ appear identical. 
This naturally extends to finite plays:
Two finite plays $p, p' \in S^+$ are $o$-indistinguishable, written $p \sim_o p'$, if $|p| = |p'|$ and for each $0 \leq i < |p|$, $p(i) \sim_o p'(i)$.
An \emph{$o$-strategy} is a function $f : S^+ \to \moves$ that cannot distinguish between $o$-indistinguishable plays, i.e., for all $p, p' \in S^+$ with $p \sim_o p'$ we have $f(p) = f(p')$. 
We denote with $\strats{\calG, o}$ the set of all $o$-strategies in $\calG$. 
Now, assume that $\stratVars = \{x, y, \ldots\}$ is a set of \emph{strategy variables}. 
We consider \SLI{} formulas that are generated by the following grammar:
\begin{align*}
	\psiSLI &:= a \mid \neg \psiSLI \mid \psiSLI \land \psiSLI \mid \ltlN \psiSLI \mid \psiSLI \ltlU \psiSLI \\
	\varphiSLI &:= \psiSLI \mid \varphiSLI \land \varphiSLI \mid \varphiSLI \lor \varphiSLI \mid \forall x^o\ldot \varphiSLI \mid \exists x^o\ldot \varphiSLI \mid \bindi{\agent{i}}{x} \varphiSLI
\end{align*}
where $a \in \ap$, $x \in \stratVars$, $\agent{i} \in \agents$, and $o \in \obs$ is an observation.
We use the usual Boolean connectives $\lor, \to, \leftrightarrow$, and constants $\top, \bot$, as well as the derived LTL operators \emph{eventually} $\ltlF \psiSLI$, \emph{globally} $\ltlG \psiSLI$, and \emph{weak until} $\psiSLI_1 \ltlW \psiSLI_2$.
Note that we do not allow nested state formulas within path formulas.
In \SLI{}, quantification ($\forall x^o$ and $\exists x^o$) ranges over a strategy with fixed observation $o$. 
The agent binding construct $\bindi{\agent{i}}{x} \varphiSLI$, then evaluates $\varphiSLI$ after binding agent $\agent{i}$ to some previously quantified strategy $x$ (cf.~\cite{MogaveroMPV14}).
The semantics of \SLI{} is defined as expected: We maintain a partial mapping $\Delta : \stratVars \to \strats{\calG}$ that maps strategy variables to strategies, and a mapping $\Theta : \agents \to  \strats{\calG}$ mapping agents to strategies. 
Whenever we evaluate $\forall x^o$ or $\exists x^o$, we quantify over a strategy in $\strats{\calG, o}$ and add it to $\Delta$; when evaluating an agent binding $\bindi{\agent{i}}{x}$, we update $\Theta$ by mapping agent $\agent{i}$ to strategy $\Delta(x)$.
Once we reach a path formula $\psi$, we check if the path resulting from the strategy profile $\Theta$ satisfies the LTL formula $\psi$. 
We give the full semantics in the \ifFull{appendix}{full version \cite{fullVersion}}.  
We write $(\calG, \{\sim_o\}_{o \in \obs}) \models_{\text{\SLI}} \varphiSLI$ if $\varphiSLI$ holds in $\calG, \{\sim_o\}_{o \in \obs}$.

\paragraph{Hyper Strategy Logic}

\HyperSL{} \cite{BeutnerF24SL} is centered around the idea of combining strategic reasoning (as possible in SL) with the ability to express hyperproperties (as possible in logics such as 
\HyperCTLS{} \cite{ClarksonFKMRS14}).
In addition to the strategy variables $\stratVars$, we assume that $\pathVars = \{\pi, \pi_1, \ldots\}$ is a set of \emph{path variables}.
Path and state formulas in \HyperSL{} are generated by the following grammar:
\begin{align*}
	\psi &:= a_\pi \mid \neg \psi \mid \psi \land \psi \mid \ltlN \psi \mid \psi \ltlU \psi \\
	\varphi &:= \forall x\ldot \varphi \mid \exists x\ldot \varphi \mid \varphiSLI \land \varphiSLI \mid \varphiSLI  \lor \varphiSLI \mid  \psi\big[\pi_k : \bX_k \big]_{k=1}^m
\end{align*}
where $a \in \ap$, $\pi, \pi_1, \ldots, \pi_m \in \pathVars$ are path variables, $x \in \stratVars$, and $\bX_1, \ldots, \bX_m : \agents \to \stratVars$ are strategy profiles that assign a strategy variable to each agent.

In \HyperSL{}, we can quantify over strategies (under full information), and evaluate LTL formulas on \emph{multiple} paths at the same time (thus expressing hyperproperties). 
Formally, $\psi[\pi_k : \bX_k ]_{k=1}^m$ expresses a hyperproperty over $m$ paths in the CGS, where the $k$th path (bound to $\pi_k$) is the unique play where each agent $\agent{i}$ plays strategy $\bX_k(\agent{i})$. 
In the LTL formula $\psi$, we then use path-variable-indexed atomic propositions: The formula $a_\pi$ holds iff AP $a$ holds on the path bound to path variable $\pi$.
\HyperSL{} thus allows us to express \emph{temporal} properties on multiple strategy profiles at the same time. 
We refer the reader to \cite{BeutnerF24SL} for details. 
The semantics of \HyperSL{} is defined as expected:
Similar to \SLI{}, we collect all strategies in a mapping $\Delta : \stratVars \to \strats{\calG}$. 
When evaluating $\psi[\pi_k : \bX_k ]_{k=1}^m$, we then define $\pi_k$ to be the unique path where each agent $\agent{i} \in \agents$ plays strategy $\Delta(\bX_k(\agent{i}))$, and evaluate the LTL formula on the resulting $m$ paths $\pi_1, \ldots, \pi_m$. 
We give a full semantics in the \ifFull{appendix}{full version \cite{fullVersion}}. 
We write $\calG \models \varphi$ if $\calG$ satisfies $\varphi$.

\section{Encoding \SLI{} Into \HyperSL}\label{sec:enc1}

Note that the underlying game structure differs between both logics (\HyperSL{} is defined on plain CGSs and \SLI{} on CGSs with an observation model), so our encodings have to modify both the formula and the underlying system. 

In this section, we show that we can encode \SLI{} into \HyperSL{} by recalling the encoding from \cite{BeutnerF24SL}.
Our encoding is based on the well-known observation that acting under imperfect information \emph{is} a hyperproperty \cite{BozzelliMP15}.

\paragraph{Injective Labeling and Action Recording}

Strategies are defined as functions $S^+ \to \moves$, and $\sim_o$ is defined as a relation on states, i.e., both are defined directly on components of the game structure.
In contrast, \emph{within} our logic, we only observe the evaluation of the atomic propositions.
In the first step, we thus modify the game structure to provide sufficient information within its atomic propositions.

\begin{restatable}{definition}{ilar}
	A CGS $\calG = (S, s_0, \moves, \kappa, L)$ is \emph{injectively labeled} (IL) if $L : S \to 2^\ap$ is injective.
	The  CGS is \emph{action recording} (AR) if for each agent $\agent{i} \in \agents$ and every action $\alpha \in \moves$,  there exists an AP $\langle \agent{i}, \alpha \rangle \in \ap$ that holds in a state iff agent $\agent{i}$ played action $\alpha$ in the last step.
\end{restatable}

We can easily modify any CGS to be IL and AR (note that \SLI{} model-checking is undecidable):

\begin{restatable}{lemma}{makeILAR}\label{lem:sli_il_ar}
	Given an \SLI{} MC instance $(\calG, \{\sim_o\}_{o \in \obs}, \varphiSLI)$ there exists an effectively computable \SLI{} instance $(\calG', \{\sim'_o\}_{o \in \obs}, \varphiSLI')$ where \textbf{(1)} $\calG'$ is IL and AR, and \textbf{(2)} $(\calG, \{\sim_o\}_{o \in \obs}) \models_{\text{\SLI{}}} \varphiSLI$ iff $(\calG', \{\sim'_o\}_{o \in \obs}) \models_{\text{\SLI{}}} \varphiSLI'$.
\end{restatable}

Now assume some fixed CGS $\calG = (S, s_0, \moves, \kappa, L)$, observation family $\{\sim_o\}_{o \in \obs}$, and \SLI{} formula $\varphiSLI$.
Using \Cref{lem:sli_il_ar}, we assume, w.l.o.g., that $\calG$ is IL and AR.

\paragraph{Enforcing Imperfect Information}

First, we construct a formula that identifies pairs of states that are indistinguishable according to $\sim_o$.
For $o \in \obs$, we define the \HyperSL{} path formula $\mathit{ind}_o$ over path variables $\pi_1, \pi_2$ as follows:
\begin{align*}
	\mathit{ind}_o := \!\!\! \bigvee_{(s, s') \in \sim_o}\!\Big(\!&\bigwedge_{a \in L(s)} a_{\pi_1} \land \!\!\! \bigwedge_{a \in \ap\setminus L(s)} \!\!\! \neg a_{\pi_1} \; \land \\
	&\quad\bigwedge_{a \in L(s')} a_{\pi_2} \land \!\! \bigwedge_{a \in \ap\setminus L(s')} \!\! \neg a_{\pi_2}\Big)
\end{align*}
It is easy to see that on any injectively labeled game structure $\mathit{ind}_o$ holds if the two paths bound to $\pi_1, \pi_2$ are $\sim_o$-related in their first state.
Given an observation $o \in \obs$, and strategy variable $x \in \stratVars$, we define a formula $\mathit{ii}_o(x)$ that holds on a strategy iff this strategy is an $o$-strategy (in all reachable situations and for all agents) as follows:
\begin{align*}
	\mathit{ii}_o(x) &:= \forall y_1, \ldots, y_n, y'_1, \ldots, y'_n. \span\span \\
	&\quad\bigwedge_{\agent{i} = 1}^n \psiSLI^{\agent{i}}_o \left[\begin{matrix}
		\pi_1 : (y_1, \ldots, y_{i-1}, x, y_{i+1}, \ldots, y_n )\\
		\pi_2 : (y'_1, \ldots, y'_{i-1}, x, y'_{i+1}, \ldots, y'_n ),
	\end{matrix}\right]
\end{align*}
where
\begin{align*}
	\psi_o^{\agent{i}} := \Big(\ltlN \bigwedge_{\alpha \in \moves} &\langle  \agent{i}, \alpha\rangle_{\pi_1} \leftrightarrow \langle \agent{i}, \alpha \rangle_{\pi_2}  \Big) \ltlW \big( \neg \mathit{ind}_o \big).
\end{align*}
The path formula $\psiSLI_o^{\agent{i}}$ compares two paths $\pi_1, \pi_2$ and states that as long as prefixes of those two paths are $o$-indistinguishable (i.e., $\mathit{ind}_o$ holds in each step), the action selected by agent $\agent{i}$ is the same on both prefixes (using the fact that the structure records actions). 
As we do not know which agents might end up playing strategy $x$, we assert that $x$ behaves as an $o$-strategy for all agents.
For each $\agent{i} \in \agents$, we thus compare two paths where $\agent{i}$ plays $x$, but all other agents play some arbitrary strategy, and assert that $\psiSLI_o^{\agent{i}}$ holds for those two paths. 
Strategy $x$ must thus respond with the same action on any two $o$-indistinguishable prefixes; in all reachable situations for all agents.

\paragraph{The Translation}

Using $\mathit{ii}_o(x)$ as a building block, we can translate \SLI{} into \HyperSL{}.
Let $\dot{\pi} \in \pathVars$ denote some \emph{fixed} path variable.
\SLI{} path formulas are then translated directly into \HyperSL{} path formulas by indexing APs with $\dot{\pi}$:
\begin{align*}
	\slToHyper{a} &:= a_{\dot{\pi}} \\
	  \slToHyper{\neg \psiSL} &:= \neg \slToHyper{\psiSL} \\
	   \slToHyper{\psiSL_1 \land \psiSL_2} &:= \slToHyper{\psiSL_1} \land \slToHyper{\psiSL_2}\\
	 \slToHyper{\ltlN \psiSL} &:= \ltlN \slToHyper{\psiSL} \\
	   \slToHyper{\psiSL_1 \ltlU \psiSL_2} &:= \slToHyper{\psiSL_1} \ltlU \slToHyper{\psiSL_2}
\end{align*}
To translate state formulas, we track the current agent binding -- which \SLI{} formalizes via agent bindings of the form $\bindi{\agent{i}}{x}$ -- using a partial strategy profile $\bX : \agents \rightharpoonup \stratVars$ (we write $\emptyset$ for the empty profile). 
We can then use $\bX$ to construct the unique path $\dot{\pi}$ whenever we encounter a path formula:
\begin{align*}
	\slToHyper{\varphiSL_1 \land \varphiSL_2}^\bX &\!:=\! \slToHyper{\varphiSL_1}^\bX \! \land \slToHyper{\varphiSL_2}^\bX\\ \slToHyper{\varphiSL_1 \lor \varphiSL_2}^\bX &\!:=\! \slToHyper{\varphiSL_1}^\bX \lor \slToHyper{\varphiSL_2}^\bX\\
	\slToHyper{\psiSL}^\bX &:= \slToHyper{\psiSL}[\dot{\pi} : \bX] \\
	\slToHyper{\bindi{\agent{i}}{x} \varphiSL}^\bX &:= \slToHyper{\varphiSL}^{\bX[\agent{i} \mapsto x]}\\
	\slToHyper{\forall x^o. \varphiSLI}^\bX &:= \forall x\ldot \mathit{ii}_o(x) \rightarrow \slToHyper{\varphiSLI}^\bX \\
	\slToHyper{\exists x^o. \varphiSLI}^\bX &:= \exists x \ldot  \mathit{ii}_o(x)  \land \slToHyper{\varphiSLI}^\bX
\end{align*}
Note that \HyperSL{} does not allow implications between state formulas, but we can convert $\varphi_1 \to \varphi_2$ to $\neg \varphi_1 \lor \varphi_2$, and push the negation into the path formulas. 

\begin{restatable}{theorem}{sliToHyper}
	For any \SLI{} instance $\big((\calG, \{\sim_o\}_{o \in \obs}), \varphi\big)$ where $\calG$ is IL and AR (cf.~\Cref{lem:sli_il_ar}), we have $(\calG, \{\sim_o\}_{o \in \obs}) \models_{\text{SL}_{\text{ii}}} \varphi$ iff $\calG \models \slToHyper{\varphiSLI}^\emptyset$.
\end{restatable}

To see the above, observe that for any strategy $f \in \strats{\calG}$ and state $s \in S$, $f$ satisfies $\mathit{ii}_o$ in state $s$ iff $f$ is an $o$-strategy in all \emph{reachable situations} starting from $s$.
Phrased differently, for any strategy $f$ that satisfies $\mathit{ii}_o$ there exists some proper $o$-strategy $f' \in \strats{\calG, o}$ that agrees with $f$ in all reachable situations. 
As any strategy will only be queried on plays that are compatible with the strategy itself, this suffices to encode the \SLI{} semantics. 

\paragraph{Encoding Size}

We can analyze the size of our formula encoding.
Formula $\mathit{ind}_o$ is of size $\calO\big(  |\ap|  \cdot |\!\!\sim_o\!\!|  \big)$, and $\mathit{ii}_o$ is of size $\calO \big( |\agents| \cdot (|\moves| + |\mathit{ind}_o| )\big)$.
The translation for path formulas is linear, and for each quantifier we add one instance of $\mathit{ii}_o$, so the size of $\slToHyper{\varphiSLI}^\emptyset$ is bounded by
\begin{align*}
	\calO \Big( |\varphi| \cdot |\agents| \cdot \big(|\moves| + |\ap| \cdot \max_{o \in \mathit{Obs}} |\!\sim_o\!| \big) \Big).
\end{align*}
Note that this assumes the worst case, where $\mathit{ind}_o$ enumerates all pairs in $\sim_o$. 
In most cases, we can identify indistinguishable states easily (by, e.g., looking at a particular AP), reducing the size of $\mathit{ind}_o$ significantly.

\section{Encoding \HyperSL{} Into \SLI{}}\label{sec:hypersl_in_sli}

In this section, we consider the reverse problem and encode a \HyperSL{} MC instance into an equivalent \SLI{} MC instance.  
Let $\calG = (S, s_0, \moves, \kappa, L)$ be a fixed CGS and $\varphi$ be a fixed \HyperSL{} formula. 
The idea of our translation is to model a hyperproperty on a self-composition of the game structure \cite{BartheDR11}.
That is, we construct a new game structure that simulates multiple copies of $\calG$ in parallel. 
Intuitively, each copy will correspond to one of the paths we construct in the \HyperSL{} formula.
We then ensure that an agent does not observe the state of other copies by using the observation available in \SLI{}.
We assume, w.l.o.g., that the set of path variables $\pathVars$ is finite.
Similar to \Cref{lem:sli_il_ar}, we make use of the fact that we can assume w.l.o.g. that the game structure is IL and AR:

\begin{restatable}{lemma}{hyperslILandAR}
	Given a \HyperSL{} MC instance $(\calG, \varphi)$, there exists an effectively computable MC instance $(\calG', \varphi')$ where \textbf{(1)} $\calG'$ is IL and AR, and \textbf{(2)} $\calG \models \varphi$ iff $\calG' \models \varphi'$.
\end{restatable}

\paragraph{Self-Composition}

We first formally define the self-composition of a game structure.
In our construction, we identify each copy in the self-composition by a path variable. 
States in the composition are, therefore, functions $\pathVars \to S$ that assign a state to each path variable.

\begin{definition}
	Define $\calG_\pathVars := (\pathVars \to S, \prod_{\pi \in \pathVars} s_0, \moves, \kappa', L')$ as the CGS over atomic propositions $\ap_\pathVars := \{ \cparai{a}{\pi} \mid a \in \ap, \pi \in \pathVars \}$ and agents $\agents_\pathVars := \{ \clanglei{\agent{i}}{\pi} \mid \agent{i} \in \agents, \pi \in \pathVars \}$ where $\kappa' : (\pathVars \to S) \times (\agents_\pathVars \to \moves) \to (\pathVars \to S)$ is defined as follows (where $\bA : \agents_\pathVars \to \moves$):
	\begin{align*}
		\kappa'\Big(\prod_{\pi \in \pathVars} s_\pi, \bA\Big) := \prod_{\pi \in \pathVars} \kappa\Big(s_\pi, \prod_{\agent{i} \in \agents} \bA (\clanglei{\agent{i}}{\pi}) \Big)
	\end{align*}
	and	$L'\big(\prod_{\pi \in \pathVars} s_\pi\big) :=\bigcup_{\pi \in \pathVars} \big\{ \cparai{a}{\pi} \mid  a \in L(s_\pi) \big\}$.
\end{definition}

In $\calG_\pathVars$, states are functions $\pathVars \to S$ (recall that $\pathVars$ and $S$ are both finite), and the initial state is the function $\prod_{\pi \in \pathVars} s_0$, mapping all path variables to the initial state $s_0$. 
Each state in $\calG_\pathVars$ records a state for every $\pi \in \pathVars$ (which we refer to as the $\pi$-copy), so each path in $\calG_\pathVars$ defines paths in $\calG$ for all path variables.
Agents in $\calG_\pathVars$ are of the form $\clanglei{\agent{i}}{\pi} \in \agents_\pathVars$, and APs are of the form $\cparai{a}{\pi} \in \ap_\pathVars$, i.e., they are indexed by path variables. 
Intuitively, for each $\pi \in \pathVars$, the agents $\{ \clanglei{\agent{i}}{\pi} \mid \agent{i} \in \agents\}$ are responsible for updating the system copy for $\pi$, so their strategies define the $\pi$-path.  
The transition function thus takes a state $\prod_{\pi \in \pathVars} s_\pi$ (where, for all $\pi \in \pathVars$, $s_\pi$ is the current state of the $\pi$-copy) and an action vector $\bA : \agents_\pathVars \to \moves$, and returns a new state $\prod_{\pi \in \pathVars} s'_\pi$.
Here, each $s'_\pi$ is defined as $\kappa\big(s_\pi, \prod_{\agent{i} \in \agents} \bA(\clanglei{\agent{i}}{\pi}) \big)$, i.e., we use $\calG$'s transition function $\kappa$ and construct the action profile $\agents \to \moves$ by mapping each agent $\agent{i} \in \agents$ (in $\calG$) to the action chosen by agent $\clanglei{\agent{i}}{\pi}$ (in $\calG_\pathVars$). 

We turn $\calG_\pathVars$ into a game structure under imperfect information by setting $\obs = \{o_\pi \mid \pi \in \pathVars\}$ and defining 
\begin{align*}
	\sim_{o_\pi} := \big\{ \big( \prod_{\pi \in \pathVars} s_\pi, \prod_{\pi \in \pathVars} s'_\pi  \big) \mid s_\pi = s'_\pi \big\}.
\end{align*}
Two states $ \prod_{\pi \in \pathVars} s_\pi$ and $\prod_{\pi \in \pathVars} s'_\pi$ thus appear indistinguishable under observation $o_\pi$ if the state of the $\pi$-copy agrees. 
A strategy under observation $o_\pi$ can thus only observe the $\pi$-copy and not base its decision on the other system copies. 

\paragraph{Translating Path Formulas}

We can now translate the \HyperSL{} formula into an equivalent \SLI{} formula over the self-composed game structure. 
The idea of our construction is to interpret path $\pi$ in the \HyperSL{} formula as the path traversed by the $\pi$-copy in $\calG_\pathVars$. 
In HyperSL path formulas, we, therefore, replace every path-variable-indexed atomic formula $a_{\pi}$ (where $a \in \ap$) with the atomic proposition $\cparai{a}{\pi} \in \ap_\pathVars$ in $\calG_\pathVars$:
\begin{align*}
	\transi{a_{\pi}} &:= \cparai{a}{\pi} \span\\
	\transi{\neg \psi} &:= \neg \transi{\psi} \\
	\transi{\psi_1 \land \psi_2} &:= \transi{\psi_1} \land \transi{\psi_2}    \\
	 \transi{\ltlN \psi} &:= \ltlN \transi{\psi}  \\
	 \transi{\psi_1 \ltlU \psi_2} &:= \transi{\psi_1} \ltlU \transi{\psi_2}
\end{align*}

\paragraph{Strategy Equality}

For state formulas, we need some additional gadget. 
In the \HyperSL{} formula, each strategy variable $x$ can be used on multiple paths. 
Yet, in $\calG_\pathVars$, each strategy only acts in a fixed copy (determined by its observation).
We, therefore, translate the quantification over a strategy variable $x$ in the \HyperSL{} formula into quantification over $|\pathVars|$-many variables $\{x_\pi \mid \pi \in \pathVars\}$ in \SLI{}. 
Each strategy variable $x_\pi$ will act in the $\pi$-copy of $\calG_\pathVars$ and thus be assigned observation $o_\pi$.
Whenever strategy variable $x$ is used in the \HyperSL{} formula to construct path $\pi \in \pathVars$, our translation uses strategy variable $x_\pi$. 
Consequently, we need to ensure that the strategies bound to $\{x_\pi \mid \pi \in \pathVars\}$ (who all will act in different copies of $\calG_\pathVars$) denote the same strategy, i.e., respond with the same action to the same prefix (in different copies of $\calG_\pathVars$).
For $x \in \stratVars$ and $\pi, \pi' \in \pathVars$, we define formula $\mathit{eq}(x_\pi, x_{\pi'})$ as follows:
\begin{align*}
	&\bigforall_{\agent{j} \in \agents, \pi \in \pathVars} y_{\agent{j},\pi}^{o_\pi}\ldot  \bigdot_{\agent{j} \in \agents, \pi \in \pathVars} \!\!\!\!\!\!\!\! \bindi{\clanglei{\agent{j}}{\pi}}{y_{\agent{j}, \pi}} \\
	&\quad\quad\quad\bigwedge_{\agent{i}_1, \agent{i}_2 \in \agents} \bindi{\clanglei{\agent{i}_1}{\pi}}{x_\pi}  \bindi{\clanglei{\agent{i}_2}{\pi'}}{x_{\pi'}} \; \psi_{\agent{i}_1, \agent{i}_2}^\mathit{eq},
\end{align*}
where $\psi_{\agent{i}_1, \agent{i}_2}^\mathit{eq}$ is defined as follows:
\begin{align*}
	\psi_{\agent{i}_1, \agent{i}_2}^\mathit{eq}:= &\Big(\ltlN \bigwedge_{\alpha \in \moves} \cparai{\langle \agent{i}_1, \alpha \rangle}{\pi} \leftrightarrow \cparai{\langle \agent{i}_2, \alpha \rangle}{\pi'} \Big) \\
	&\quad\quad\quad\quad\ltlW \Big( \bigvee_{a \in \ap} \cparai{a}{\pi} \not\leftrightarrow \cparai{a}{\pi'}  \Big).
\end{align*}
Here, we write $\!\bigdot\!$ as an abbreviation for the concatenation of multiple agent bindings in \SLI{}.

Formula $\mathit{eq}(x_\pi, x_{\pi'})$ expresses that $x_\pi$ and $x_{\pi'}$ denote the same strategy, even though $x_\pi$ operates in the $\pi$-copy and $x_{\pi'}$ operates in the $\pi'$-copy.
The underlying idea is similar to the one used in \Cref{sec:enc1}: we use the fact that $\calG$ is IL and AR and can thus reason about the action selection of agents within our logic. 
In $\mathit{eq}(x_\pi, x_{\pi'})$, we consider every agent $\agent{j} \in \agents$ and every path variable $\pi \in \pathVars$, quantify universally over a strategy $y_{\agent{j}, \pi}$ with observation $o_\pi$, and then bind agent $\clanglei{\agent{j}}{\pi}$ in $\calG_\pathVars$ to strategy $y_{\agent{j}, \pi}$.
We then consider any pair of agents $\agent{i}_1, \agent{i}_2 \in \agents$, and want to express that if $\agent{i}_1$ plays $x_\pi$ in the $\pi$-copy (so agent $\clanglei{\agent{i}_1}{\pi}$ plays $x_\pi$) and $\agent{i}_2$ plays $x_{\pi'}$ in the $\pi'$-copy (so agent $\clanglei{\agent{i}_2}{\pi'}$ plays $x_{\pi'}$), the same prefix (in the $\pi$ and $\pi'$ copies) results in the same action selected by $\clanglei{\agent{i}_1}{\pi}$ and $\clanglei{\agent{i}_2}{\pi'}$.
We thus re-bind agents $\clanglei{\agent{i}_1}{\pi}$ and $\clanglei{\agent{i}_2}{\pi'}$ to strategies $x_\pi$ and $x_{\pi'}$, respectively. 
The resulting strategy profile now explores all possible reachable situations in which agent $\agent{i}_1$ in the $\pi$-copy plays strategy $x_\pi$, and agent $\agent{i}_2$ in the $\pi'$-copy plays strategy $x_{\pi'}$.
The path formula $\psi_{\agent{i}_1, \agent{i}_2}^\mathit{eq}$ then asserts that as long as the prefix in the $\pi$ and $\pi'$ copy is equal (using the fact that $\calG$ is IL), the action selected by $\clanglei{\agent{i}_1}{\pi}$ and $\clanglei{\agent{i}_2}{\pi'}$ is the same (using the fact that $\calG$ is AR). 
Note that we assumed that $\calG$ is AR and thus includes APs of the form $\langle \agent{i}, \alpha \rangle$ for each agent $\agent{i}$ and action $\alpha$, so $\calG_\pathVars$ includes APs of the form $\cparai{\langle \agent{i}, \alpha \rangle}{\pi}$.

\paragraph{Translating State Formulas}

Using the $\mathit{eq}(\cdot, \cdot)$ construction, we can translate \HyperSL{} state formulas:
\begin{align*}
	\transi{\forall x\ldot \varphi} &:= \bigforall_{\pi \in \pathVars} x_\pi^{o_\pi}\ldot \Big(\bigwedge_{\pi, \pi' \in \pathVars} \mathit{eq}(x_\pi, x_{\pi'}) \Big)\to \transi{\varphi}\\ 
	\transi{\exists x\ldot \varphi} &:= \bigexists_{\pi \in \pathVars} x_\pi^{o_\pi}\ldot  \Big(\bigwedge_{\pi, \pi' \in \pathVars} \mathit{eq}(x_\pi, x_{\pi'})\Big) \land \transi{\varphi}\\ 
	\transi{\varphi_1 \land \varphi_2} &:= \transi{\varphi_1}\land \transi{\varphi_2} \\
	\transi{\varphi_1 \lor \varphi_2} &:= \transi{\varphi_1}  \lor  \transi{\varphi_2}\\
	\transi{\psi\big[\pi_k : \bX_k \big]_{k =1}^m} &:=\Big(\!\!\!\! \bigdot_{\agent{i} \in \agents, k=1}^m \!\!\!\!  \bindi{\clanglei{\agent{i}}{\pi_k}}{(\bX_k(\agent{i}))_{\pi_k}} \Big) \; \transi{\psi}
\end{align*}
Here, we again write $\!\bigdot\!$ as an abbreviation for the concatenation of multiple agent bindings in \SLI{}.
Whenever we translate quantification over variable $x$, we instead quantify over $|\pathVars|$-many variables $\{x_\pi \mid \pi \in \pathVars\}$  with the appropriate observation and make sure that they all denote the same strategy by using $\mathit{eq}(\cdot, \cdot)$.
When we translate $\psi\big[\pi_k : \bX_k \big]_{k=1}^m$, we fix strategies for all agents $\cparai{\agent{i}}{\pi} \in \agents_\pathVars$, and evaluate the translated path formula $\transi{\psi}$.
The key idea here is to bind each agent $\cparai{\agent{i}}{\pi_k}$ to the strategy that corresponds to the strategy $\bX_k(\agent{i})$ in the \HyperSL{} formula, i.e., the strategy that takes the role of agent $\agent{i}$ in the $\pi_k$-copy of $\calG_\pathVars$ (i.e., the strategy of $\cparai{\agent{i}}{\pi_k}$) should equal the strategy that agent $\agent{i}$ uses to construct path $\pi_k$ in the \HyperSL{} formula.
In our encoding, we translate each strategy quantifier over $x$ in \HyperSL{}, into $|\pathVars|$-many strategies $\{x_\pi \mid \pi \in \pathVars\}$ in \SLI{}, where each $x_\pi$ acts in the $\pi$-copy of $\calG_\pathVars$. 
Consequently, the strategy variable that corresponds to $\bX_k(\agent{i})$ in the $\pi_k$-copy is the variable $(\bX_k(\agent{i}))_{\pi_k}$.
For example, if $\bX_k(\agent{i}) = z$, we bind agent $\clanglei{\agent{i}}{\pi_k}$ to strategy variable $z_{\pi_k}$.

\begin{example}
	Consider the \HyperSL{} formula 
	\begin{align*}
		\exists x, y\ldot \forall z. ( \ltlG (a_{\pi_1} \to b_{\pi_2}) )[\pi_1 : (x, y), \pi_2 : (z, x)]
	\end{align*}
	Using our translation, and after removing non-needed variables, we obtain the following formula on $\calG_{\{\pi_1, \pi_2\}}$:
		\begin{align*}
			&\exists x_{\pi_1}^{o_{\pi_1}}, x_{\pi_2}^{o_{\pi_2}}, y_{\pi_1}^{o_{\pi_1}}. \forall z_{\pi_2}^{o_{\pi_2}}\ldot \mathit{eq}(x_{\pi_1}, x_{\pi_2}) \, \land \\
			&\quad \bindi{\clanglei{\agent{1}}{\pi_1}}{x_{\pi_1}} \bindi{\clanglei{\agent{2}}{\pi_1}}{y_{\pi_1}}  \bindi{\clanglei{\agent{1}}{\pi_2}}{z_{\pi_2}} \bindi{\clanglei{\agent{2}}{\pi_2}}{x_{\pi_2}} \\
			&\quad\quad\ltlG (\cparai{a}{\pi_1} \to \cparai{b}{\pi_2}).
		\end{align*}
\end{example}

\begin{restatable}{theorem}{hyperslToSLI}
	Let $(\calG, \varphi)$ be a \HyperSL{} instance, where $\calG$ is IL and AR. 
	Then $\calG \models \varphi$ iff $(\calG_{\pathVars}, \{\sim_{o_\pi}\}_{o_\pi \in \obs}) \models_{\text{\SLI{}}} \transi{\varphi}$.
\end{restatable}

\paragraph{Encoding Size}

We can, again, analyze the size of our encoding. 
The size of $\mathit{eq}(\cdot, \cdot)$ is of order
$$
\calO \big( |\agents| \cdot |\pathVars| + |\agents|^2 \cdot (|\moves| + |\ap|) \big).
$$
The translation of path formulas is linear, for each quantifier we add $|\pathVars|^2$-many $\mathit{eq}(\cdot, \cdot)$ constraints, and for each nested path formula we add $|\pathVars|\cdot|\agents|$-many agent bindings. 
The overall size of $\transi{\varphi}$ is thus bounded by
$$
\calO \Big( |\varphi| \cdot |\pathVars|^3 \cdot |\agents|^2 \cdot (|\moves| + |\ap|) \Big).
$$
As expected for a self-composition, $\calG_\pathVars$ has $\calO\big( |S|^{|\pathVars|} \big)$ states \cite{BartheDR11}.  

\section{Related Work}

There exist many extensions of ATL$^*$ and SL to reason about imperfect information \cite{BelardinelliLM19,DimaT11,JamrogaMM19,abs-1908-02488,BelardinelliLMR17,BerthonMM17,HuangM14,HuangM18} or hyperproperties \cite{BeutnerF21,BeutnerF23LMCS,BeutnerF24,BeutnerF24SL}.
The connection between knowledge and hyperproperties has been studied extensively \cite{Rabe16,coenen2020hierarchy,BeutnerFFM23,BeutnerF25a,BeutnerF25b}.
The work most closely related to ours is the study of \cite{BozzelliMP15}, who show that LTL$_K$ and a fragment of HyperCTL* are equally expressive (on standard Kripke structures).
We study the relation of hyperproperties and knowledge in the setting of MASs and strategies.
For example, unlike LTL$_K$, \SLI{} does not use a knowledge operator that we can directly encode. Instead, \SLI{} quantifies over strategies with a given observation relation.

\section{Conclusion}

In this paper, we have established the first formal connection between imperfect information and hyperproperties in the context of strategy logic. 
Our results shed new light on the intricate connection between knowledge and hyperproperties and allow for the transfer of tools and techniques.
For future work, it is interesting to check how existing verification tools for imperfect information \cite{JamrogaKKM19,LomuscioQR09,KurpiewskiJK19} perform on the encoding from \HyperSL{}, and if the decidable fragments of \HyperSL{} correspond to interesting classes of \SLI{} properties.

\section*{Acknowledgments}
This work was supported by the European Research Council (ERC) Grant HYPER (101055412), and by the German Research Foundation (DFG) as part of TRR 248 (389792660).

\bibliographystyle{kr}
\bibliography{references}

\iffullversion

\appendix

\section{Notation}

Given a set $X$, we write $X^+$ (resp.~$X^\omega$) for the set of non-empty finite (resp.~ infinite) sequences over $X$.
For $u \in X^\omega$ and $j \in \nat$, we write $x(j)$ for the $i$th element, $u[0,j]$ for the finite prefix up to position $j$ (of length $j + 1$), and $u[j, \infty]$ for the infinite suffix starting at position $j$.

\section{CGSs}

Recall that a CGS is a tuple $\calG = (S, s_0, \moves, \kappa, L)$ where $S$ is a finite set of states, $s_0 \in S$ is an initial state, $\moves$ is a finite set of actions, $\kappa : S \times (\agents \to \moves)\to S$ is a transition function, and $L : S \to 2^\ap$ is a labeling function.
The transition function takes a state $s$ and an \emph{action profile} $\vec{\alpha} : \agents \to \moves$ (mapping each agent an action) and returns a unique successor state $\kappa(s, \vec{\alpha})$.
A strategy in $\calG$ is a function $f : S^+ \to \moves$, mapping finite plays to actions.
We denote the set of all strategies in $\calG$ with $\strats{\calG}$.
A \emph{strategy profile} $\prod_{\agent{i} \in \agents} f_\agent{i}$ assigns each agent $\agent{i}$ a strategy $f_\agent{i} \in \strats{\calG}$.
Given strategy profile $\prod_{\agent{i} \in \agents} f_\agent{i}$ and state $s \in S$, we can define the unique path $\play_\calG(s, \prod_{\agent{i} \in \agents} f_\agent{i}) \in S^\omega$ resulting from the interaction between the agents.
We define $\play_\calG(s, \prod_{\agent{i} \in \agents} f_\agent{i})$ as the unique path $p \in S^\omega$ such that $p(0) = s$ and for every $j \in \nat$ we have 
\begin{align*}
	p(j+1) = \kappa\big(p(j), \prod_{\agent{i} \in \agents} f_\agent{i}(p[0,j]) \big).
\end{align*}
That is, in every step, we construct the action profile $\prod_{\agent{i} \in \agents} f_\agent{i}(p[0,j])$ in which each agent $\agent{i}$ plays the action determined by $f_\agent{i}$ on the current prefix $p[0,j]$.

\section{\SLI{} Semantics}

Assume $\calG = (S, s_0, \moves, \kappa, L)$ is a fixed CGS.
Given a path $p \in S^\omega$ we define the semantics of path formulas as expected:
\begin{align*}
	p &\models a &\text{iff } \quad&a \in L(p(0))\\
	p&\models \psiSL_1 \land \psiSL_2 &\text{iff } \quad &p \models \psiSL_1 \text{ and } p \models \psiSL_2\\
	p &\models \neg \psiSL &\text{iff } \quad &p \not\models \psiSL\\
	p &\models \ltlN \psiSL &\text{iff } \quad &p[1, \infty] \models \psiSL\\
	p &\models \psiSL_1 \ltlU \psiSL_2 &\text{iff } \quad &\exists j \in \nat \ldot p[j, \infty] \models \psiSL_2 \text{ and } \\
	&  \quad\quad\quad\quad\quad\quad\quad\quad\quad\quad \forall 0 \leq k < j\ldot p[k, \infty] \models \psiSL_1 \span \span
\end{align*}
In the semantics of state formulas, we keep track of a strategy for each strategy variable via a (partial) strategy assignment $\Delta : \stratVars \rightharpoonup \strats{\calG}$.
We write $\{\}$ for the unique strategy assignment with an empty domain.
As \SLI{} works with explicit agent bindings, we also keep track of a strategy for each agent using a (partial) function $\Theta : \agents \rightharpoonup \strats{\calG}$. 
We can then define:
\begin{align*}
	&s, \Delta, \Theta \models \forall x^o \ldot\varphiSL &\text{iff} \quad &\forall f \in\strats{\calG, o} \ldot \\
	& &s, \Delta[x \mapsto f], \Theta \models \varphiSL \span \span\\
	&s, \Delta, \Theta \models \exists x^o \ldot\varphiSL &\text{iff} \quad &\exists f \in\strats{\calG, o} \ldot \\
	& &s, \Delta[x \mapsto f], \Theta \models \varphiSL \span\span\\
	&s, \Delta, \Theta \models \bindi{\agent{i}}{x}\varphiSL &\text{iff} \quad &s, \Delta, \Theta[\agent{i} \mapsto \Delta(x)] \models \varphiSL\\
	&s, \Delta, \Theta \models \varphiSL_1 \land  \varphiSL_2 &\text{iff} \quad &s, \Delta, \Theta \models \varphiSL_1 \text{ and } s, \Delta, \Theta \models \varphiSL_2\\
	&s, \Delta, \Theta \models \varphiSL_1 \lor  \varphiSL_2 &\text{iff} \quad &s, \Delta, \Theta \models \varphiSL_1 \text{ or } s, \Delta, \Theta \models \varphiSL_2\\
	&s, \Delta, \Theta \models \psiSL &\text{iff} \quad &\play\big(s, \prod_{\agent{i} \in \agents} \Theta (\agent{i}) \big) \models \psiSL
\end{align*}
Strategy quantification updates the binding in $\Delta$, whereas strategy binding updates the assignment of agents in $\Theta$. 
For each path formula we use the strategy profile $\Theta$ (mapping each agent to a strategy) to construct the path on which we evaluate $\psi$.
Given a game structure $\calG$, a family $\{\sim_o\}_{o \in \obs}$, and an \SLI{} formula $\varphiSLI$, we write $(\calG, \{\sim_o\}_{o \in \obs}) \models_{\text{\SLI}} \varphiSLI$ if $s_0, \{\}, \{\} \models \varphiSLI$ in the \SLI{} semantics. 
See \cite{BerthonMMRV17,BerthonMMRV21} for concrete examples of \SLI{}.

\section{\HyperSL{} Semantics}

We fix a game structure $\calG = (S, s_0, \moves, \kappa, L)$.
A \emph{strategy assignment} is a partial mapping $\Delta : \stratVars \rightharpoonup \strats{\calG}$.
We write $\{\}$ for the unique strategy assignment with an empty domain.
In \HyperSL{}, a path formula $\psi$ refers to propositions on multiple path variables. 
We evaluate it in the context of a \emph{path assignment} $\Pi : \pathVars \rightharpoonup S^\omega$ mapping path variables to paths, similar to the semantics of \HyperCTLS{} \cite{ClarksonFKMRS14}.
Given $j \in \nat$, we define $\Pi[j, \infty]$ as the shifted assignment defined by $\Pi[j, \infty](\pi) := \Pi(\pi)[j, \infty]$.
For a path formula $\psi$, we then define the semantics in the context of path assignment $\Pi$:
\begin{align*}
	\Pi &\models_\calG a_\pi &\text{iff } \quad&a \in L\big(\Pi(\pi)(0)\big)\\
	\Pi &\models_\calG \psi_1 \land \psi_2 &\text{iff } \quad &\Pi \models_\calG \psi_1 \text{ and } \Pi \models_\calG \psi_2\\
	\Pi &\models_\calG \neg \psi &\text{iff } \quad &\Pi \not\models_\calG \psi\\
	\Pi &\models_\calG \ltlN \psi &\text{iff } \quad &\Pi[1, \infty] \models_\calG \psi\\
	\Pi &\models_\calG \psi_1 \ltlU \psi_2 &\text{iff } \quad &\exists j \in \nat\ldot \Pi[j, \infty] \models_\calG \psi_2 \text{ and } \\
	&  \quad\quad\quad\quad\quad\quad\quad\quad\quad\quad \forall 0 \leq k < j\ldot \Pi[k, \infty] \models_\calG \psi_1 \span \span
\end{align*}
The semantics for path formulas synchronously steps through all paths in $\Pi$ and evaluate $a_\pi$ on the path bound to $\pi$.
That is $a_\pi$ holds if the current state of the path assigned to $\pi$ (i.e., $\Pi(\pi)(0)$) satisfies $a$, i.e., $a \in L\big(\Pi(\pi)(0)\big)$.
State formulas are evaluated in a state $s \in S$ and strategy assignment $\Delta$ as follows:
\begin{align*}
	&s, \Delta \models_\calG \forall x \ldot\varphi  &&\text{iff }\\
	& \quad\quad\quad \quad\quad\forall f \in\strats{\calG} \ldot s, \Delta[x \mapsto f] \models_\calG \varphi \span\span\\
	&s, \Delta \models_\calG \exists x \ldot\varphi  &&\text{iff }\\
	 & \quad\quad\quad \quad\quad\exists f \in\strats{\calG} \ldot s, \Delta[x \mapsto f] \models_\calG \varphi \span\span\\
	 &s, \Delta \models \varphi_1 \land \varphiSL_2 &&\text{iff } s, \Delta \models \varphi_1 \text{ and } s, \Delta \models \varphi_2 \\
	 &s, \Delta \models \varphi_1 \lor \varphiSL_2 &&\text{iff } s, \Delta \models \varphi_1 \text{ or } s, \Delta \models \varphi_2 \\
	&s, \Delta \models_\calG \psi\big[\pi_k : \bX_k \big]_{k=1}^m &&\text{iff } \\
	& \quad\quad\quad \quad\quad\Big[\pi_k\! \mapsto \! \play_\calG\Big(s,\! \prod_{\agent{i} \in \agents} \!\!\Delta(\bX_k(\agent{i})) \Big) \Big]_{k=1}^m \models_\calG \psi \span \span
\end{align*}
To resolve a formula $\psi\big[\pi_k : \bX_k \big]_{k=1}^m$, we construct $m$ paths (bound to $\pi_1, \ldots, \pi_m$), and evaluate $\psi$ in the resulting path assignment. 
The $k$th path (bound to $\pi_k$) is the play where each agent $\agent{i}$ plays strategy $\Delta(\bX_k(\agent{i}))$, i.e., the strategy currently bound to the strategy variable $\bX_k(\agent{i})$.
We write $\calG \models \varphi$ if $s_0, \{\} \models_\calG \varphi$, i.e., the initial state satisfies state formula $\varphi$.

\section{IL and MR}

Recall the definition of IL and AR:

\ilar*

\makeILAR*
\begin{proof}
	Assume $\calG = (S, s_0, \moves, \kappa, L)$.
	Define $\ap' := \ap \uplus \{\langle \agent{i}, a \rangle  \mid \agent{i} \in \agents, \alpha \in \moves\}$.
	We then define
	\begin{align*}
		\textstyle\calG' = (S \times (\agents \to \moves), (s_0, \prod_{\agent{i} \in \agents} a), \moves, \kappa', L')
	\end{align*}
	where $\alpha \in \moves$ is some arbitrary action (in the initial state, we do not need to track the last played action). 
	For an action profile $\prod_{\agent{i} \in \agents} \alpha_\agent{i}$, we define $\kappa'$ and $l'$ by
	\begin{align*}
		\kappa' \big((s, \_), \prod_{\agent{i} \in \agents} \alpha_\agent{i}\big) &:= (\kappa(s, \prod_{\agent{i} \in \agents} \alpha_\agent{i}), \prod_{\agent{i} \in \agents} \alpha_\agent{i})\\
		L'(s, \prod_{\agent{i} \in \agents} \alpha_\agent{i}) &:= L(s) \uplus \{ \langle \agent{i}, \alpha_\agent{i} \rangle \mid \agent{i} \in \agents\}.
	\end{align*}
	The idea behind $\calG'$ is that we record the action profile that was last used in the second component of each state. 
	In each transition, we ignore the action profile in the current step, and record the new action profile n the second component. 
	In the labeling function, we can then use the action profile $\prod_{\agent{i} \in \agents} \alpha_\agent{i}$ in the second component to set the APs $\langle \agent{i}, \alpha_\agent{i} \rangle$ for all $\agent{i} \in \agents$.
	We define $\{\sim'_o\}_{o \in \obs}$ by 
	\begin{align*}
		\sim'_o := \Big\{ \big((s, \_), (s', \_)\big) \mid s \sim_o s' \Big\}.
	\end{align*}
	That is, for any observation cannot distinguish states based on the second position. 
	In particular note that $(s, \_)$ and $(s, \_)$ are always indistinguishable, i.e., all states we expanded that we added are indistinguishable under the new observation. 
	
	It is easy to see that $\calG'$ is AR and that
	$(\calG, \{\sim_o\}_{o \in \obs}) \models_{\text{\SLI{}}} \varphiSLI$ iff $(\calG', \{\sim'_o\}_{o \in \obs}) \models_{\text{\SLI{}}} \varphiSLI$.
	Note that we did not change the formula.
	In a second step, we can ensure that $\calG'$ is also IL by simply adding sufficiently many new propositions. 
	As those new propositions are never used in $\varphiSL$ so the \SLI{} semantics is unchanged.
\end{proof}

\hyperslILandAR*
\begin{proof}
	We use the same game structure $\calG$ that we constructed in the proof of \Cref{lem:sli_il_ar}.
	Similar to the construction of $\mathit{ii}_o(x)$ in \Cref{sec:enc1} we can express that a strategy should behave the same if two prefixes that denote the same prefix in $\calG$ and only differ in the second component that records the move that was last played. 
	If we restrict each quantification to strategies that do this, we recover the same semantics as in $\calG$.
\end{proof}

\fi

\end{document}